\documentclass[11pt, a4paper]{article}
\usepackage[utf8]{inputenc}
\usepackage{amsmath,amssymb,amsthm}
\usepackage{graphicx}
\usepackage{hyperref}
\usepackage{cite}
\usepackage{authblk}
\usepackage{caption} %

\usepackage{xspace,xcolor}
\usepackage{enumitem}
\usepackage[margin=2.5cm]{geometry}

\newtheorem{theorem}{Theorem}[section]

\newtheorem{lemma}[theorem]{Lemma}
\newtheorem{definition}[theorem]{Definition}
\newtheorem{remark}[theorem]{Remark}

\usepackage[linesnumbered,ruled,vlined]{algorithm2e}

\newcommand{\EQ}{Belief-Independent Sequential Equilibrium}

\title{Distributed MIS Algorithms for Rational Agents using Games}

\author{Nithin Salevemula}
\author{Shreyas Pai}
\affil{Indian Institute of Technology Madras\footnote{Email: \url{nithinprakash9999@gmail.com}, \url{shreyas@cse.iitm.ac.in}}}

\newif\ifdraft%
\drafttrue

\begin{document}

\maketitle

\begin{abstract}
We study the problem of computing a Maximal Independent Set (MIS) in distributed networks, where each node is a rational agent that receives a payoff depending on whether it is included in the MIS. In classical distributed computing, it is typically assumed that nodes follow the prescribed algorithm faithfully. However, this assumption fails when nodes are rational agents whose utilities depend on the algorithm's output. In such cases, nodes may deviate from the algorithm if it increases their expected payoff.

Classical solutions for MIS assume that nodes generate random bits honestly or rely on unique identifiers to break symmetry. However, in rational settings, nodes may manipulate randomness to gain a strategic advantage, and relying solely on unique identifiers can result in unfairness, where some nodes have zero probability of joining the MIS and thus no incentive to participate. To address these challenges, we propose two algorithms that work under a utility model, where agents are incentivized to compute locally correct solutions while also exhibiting preferences among these solutions. In these algorithms, randomness is generated through interactions between neighboring nodes, which can be viewed as simple games, where no single node can unilaterally change the outcome. This approach allows us to break symmetry while being compatible with rational behavior.

For both algorithms, we show that regardless of the execution history that has occurred, no agent can improve its expected utility by deviating from that stage, provided no other agents deviate. This is a much stronger guarantee compared to Trembling Hand Perfect Equilibrium, which is typically used in such scenarios. Both algorithms guarantee that when all the nodes follow the algorithm, every node has a positive probability of joining the MIS, and that the final output is a correct Maximal Independent Set. Finally, for both algorithms, we can guarantee termination in $ O(\log n) $ rounds with high probability under mild additional assumptions, where \(n\) is the number of nodes in the network.
\end{abstract}

\section{Introduction}
\label{sec:intro}

In many classical distributed computing settings, some fraction of the nodes is assumed to be byzantine while the rest are obedient. An obedient node follows the prescribed algorithm without deviation, whereas byzantine nodes are assumed to be controlled by an adversary that is actively trying to ensure that the algorithm's goals are not achieved. Such a setup is valid as long as the obedient nodes follow the prescribed algorithm perfectly and the fraction of byzantine nodes is not too high. However, in many real-world scenarios, the situation is a bit more nuanced: nodes in the network often have their own incentives, which makes algorithm design challenging even without the presence of a byzantine adversary. In this work, we take a "middle path" of assuming that each node in the network is a selfish agent with its own utility function that it wants to maximize. If we allow for arbitrary utility functions, the behavior of the rational nodes will be similar to that of a byzantine node\footnote{Albeit in this case, all the byzantine nodes may not be controlled by a single adversary}. Therefore, we restrict our attention to utility functions where nodes are incentivized to compute a correct solution, but may still have a preference over which correct solution they want to compute. See Section \ref{sec:model} for a full description of our model.

Consider the Maximal Independent Set (MIS) problem as a working example. Although all nodes are primarily interested in ensuring that the MIS is computed correctly, they also prefer to be included in the set. As a result, a node may deviate from the prescribed algorithm if doing so increases its chance of joining the MIS. This behavior is referred to as \emph{rationality} in the game-theoretic literature, where agents (nodes) act to maximize their individual utility. Deviations due to individual rationality can harm both the correctness and termination guarantees of classical MIS algorithms. Therefore, we need to design new algorithms that are resilient to rational behavior. Such algorithms guarantee that nodes, acting in their own self-interest, still prefer to follow the algorithm rather than deviate.

\paragraph*{Challenges for Computing MIS with Rational Agents}
\label{sec:challenges}
Classical distributed algorithms for computing a Maximal Independent Set (MIS), such as Luby's algorithm~\cite{Luby1986} and the algorithm by Métivier et al.~\cite{Metivier}, are not suitable in rational settings where nodes may deviate in order to maximize their own utility. As an illustrative example, let us examine both these algorithms under the utility model where nodes are incentivized to be included in the MIS. 

In Luby's algorithm \cite{Luby1986}, each node \(i\) proposes to join the MIS with probability \( 1/2d(i)\), where \( d(i) \) is the degree of the node \( i \). A node successfully joins the MIS if none of its neighbors propose in the same round. However, this approach fails in a rational setting: since the decision to propose is determined by internal randomness, a rational node can always falsely claim that it has proposed to join the MIS with probability \( 1/2d(i)\), but it got lucky in every round. If all nodes act this way, no node ever joins the MIS, and the algorithm never terminates. In the algorithm of Métivier et al.~\cite{Metivier}, each node generates a rank uniformly at random from the interval \([0,1]\). A node joins the MIS if its rank is lower than all of its neighbors' ranks. Again, this algorithm will not work if nodes prefer joining the MIS: a rational node can simply report a rank of 0 in every round, ensuring it always appears eligible to join. If all nodes behave similarly, no node is eliminated, no progress is made, and the algorithm stalls indefinitely.

Given the above discussion, one may argue that we should then just focus on deterministic MIS algorithms, as there has been impressive progress on this front \cite{GhaffariGFOCS24}. However, deterministic algorithms have the following problem: a node may realize that it can never join the MIS by following the algorithm, even though none of its neighbors have currently joined the MIS. In such a situation, it will have no incentive to participate further in the algorithm. As a simple example, consider a deterministic algorithm where, in each round, a node joins the MIS if its ID is smaller than those of all its neighbors. If a node $i$ has a neighbor $j$ with the globally minimum ID, $i$ knows it can never join the MIS even though $j$ has not yet joined the MIS. Therefore, $i$ has no incentive to participate in the algorithm. In fact, any deterministic algorithm would face the same problem when nodes are rational agents aiming to maximize their own utility. This problem even holds for algorithms that have a randomized part followed by a deterministic part.

These examples demonstrate that classical MIS algorithms do not translate directly to settings with rational agents. Therefore, new algorithms must be designed that are resilient to deviations that arise from rational behavior and guarantee participation from every node.

\paragraph*{Our Contributions}
\label{sec:contributions}
We use the classical framework of \emph{extensive-form games} to analyze the MIS problem in distributed networks with rational agents. 
Each node in the network is modeled as a rational player, motivated by the incentive to be included in the MIS. 
We define a \emph{strategy algorithm} that specifies, for each node, the action it should take in every round until termination, 
as a function of its state—including those that may arise due to deviations from the prescribed algorithm, either by itself or by other nodes. 
We then propose two such strategy algorithms in which nodes have no unilateral incentive to deviate at any state, 
even at states that occur as a result of earlier deviations. 
We refer to this property as a \emph{belief-independent sequential equilibrium}, 
which provides a stronger guarantee than Nash equilibrium, sequential equilibrium, or trembling-hand perfect equilibrium in games of imperfect information.

Furthermore, when all agents follow the strategy algorithm, the nodes joining the independent set form a valid MIS, and every node has a non-zero probability of being included in the MIS. The following is a brief description of our proposed strategy algorithms.
\begin{itemize}
    \item The first algorithm (Section \ref{sec:mis-rps}): In each iteration, every active node plays a simple two-player game of rock-paper-scissors with each of its neighbors in order to decide who gets to join the MIS. Only a node that wins all its games against its neighbors in an iteration gets to join the MIS. This rule allows us to ensure that no agent can benefit from unilaterally deviating from the algorithm at any point. Moreover, when all nodes follow the algorithm, we can guarantee termination in \(O(\log n)\) rounds with high probability, 
assuming the maximum degree of the network graph is constant, where \(n\) is 
the number of nodes in the network.
    
    \item The second algorithm (Section \ref{sec:mis-crypto}) uses lightweight cryptographic assumptions to generate randomness in a way that prevents manipulation. Each agent computes its priority for the current round by combining its own random value with a signed random value received from one of its neighbors. Since the randomness is jointly determined, no agent can unilaterally bias the outcome. These priorities are then used to simulate one round of the algorithm of Métivier et al.~\cite{Metivier}. Repeating this process, the algorithm guarantees termination in \( O(\log n) \) rounds with high probability.
\end{itemize}

\section{The LOCAL Model with Rational Agents}
\label{sec:model}
We work in the LOCAL model, a classic synchronous message-passing model of distributed computing \cite{Linial1987}. The network is abstracted as \(G = (V, E)\), an undirected graph, where \(V\) is the set of nodes (agents) and \(E \subseteq V \times V\) is the set of edges. As usual we denote \(n = |V|\) and \(m = |E|\).  
Each node represents a selfish agent that acts to maximize its own utility, and each edge \((i, j) \in E\) indicates that agents \(i\) and \(j\) can communicate directly. We assume no prior communication occurs between the agents before the algorithm begins. Each agent is expected to irrevocably output one of three values: \(1\), declaring itself as a member of the MIS; \(0\), declaring itself to not belong to the MIS; or \(\bot\), indicating an abort. The output of a node, once computed, is visible to all the neighbors of that node.

The utility of each node depends only on its own output and those of its neighbors. As discussed earlier, if we allow for arbitrary utilities, the behavior of the rational nodes will be similar to that of Byzantine nodes, which will lead to trivial impossibilities. Therefore, we consider a particular family of utility functions in which nodes are incentivized to compute a locally correct solution. Additionally, each node $i$ gets utility value \(v_i > 0\) if it outputs \(1\) and all its neighbors output \(0\). More precisely, the utility function for a node \( i \) can be defined as a function on its own output and its neighbors' output. Let \( \text{out}(i) \in \{1,0,\bot\} \) denote the output of node \(i\) and let \( \text{out}(N(i))\) denote the list of outputs of all its neighbors. Then, the utility function \(u_i\) for node \(i\) is defined in Equation~\ref{eq:utility}. Here we slightly abuse notation by interpreting \( \text{out}(N(i))\) as a set in the case analysis.  
The conditions are evaluated in order; once a condition is satisfied, the subsequent conditions are not considered, even though they may be applicable.
\begin{figure}
\centering
\[
u_i\bigl(\text{out}(i), \text{out}(N(i))\bigr) =
\begin{cases}
0, &
\text{if } 
\text{out}(i) = \bot 
~\text{or}~ 
\bot \in \text{out}(N(i))
\quad (\text{locally~aborted}), \\[6pt]
0, &
\text{if } 
\text{out}(i) = 0 
~\text{and}~ 
1 \in \text{out}(N(i))
\quad (\text{locally~valid, } i \text{ not~in~MIS}), \\[6pt]
v_i, &
\text{if } 
\text{out}(i) = 1 
~\text{and}~ 
1,\bot \notin \text{out}(N(i))
\quad (\text{locally~valid, } i \text{ in~MIS}), \\[6pt]
-\infty, & \text{otherwise}
\quad (\text{locally invalid solution}).
\end{cases}
\]
\caption{Local utility function for node \(i\).}
\label{eq:utility}
\end{figure}

Given the above model and assumptions, our goal is to design a distributed algorithm in which nodes, despite being able to deviate, find it optimal to follow the algorithm. However, if
a deviation does occur, we cannot simply ignore the deviating node; the algorithm must still
prescribe how it should optimally behave from that point onward. In other words, the algorithm
should tell each node what to do in every situation it could possibly face, even those that would occur only if some of them deviated earlier. This is precisely analogous to the notion of a \emph{strategy} in game theory.
To capture this requirement, we introduce the notion of a \emph{strategy algorithm}. Intuitively, it specifies how a node should act in every possible situation it may encounter, including those that arise due to deviations by itself or its neighbors. To define this formally, we first define the notion of a \emph{state}, and then use it to precisely characterize what constitutes a strategy algorithm.

\begin{definition}[State]\label{def:local_history} The \emph{state} of a node \( i \) at the beginning of round \( t \) consists of all messages sent and received by \( i \) in rounds \( t' < t \), the values of its local variables (including any private randomness), and any observable outputs of its neighbors from earlier rounds. \end{definition}

\begin{definition}[Strategy Algorithm]\label{def:strategy_algorithm}
A \emph{strategy algorithm} in the LOCAL model with rational nodes is a special type of algorithm, which 
specifies, for each node \( i \), the action it should take in every round until termination 
as a function of its state, including states that may arise due to deviations 
from the prescribed algorithm by itself or by other nodes.
\end{definition}

\begin{remark}
This means that in every round, all nodes know the precise internal computation and the external action prescribed by the algorithm. However, whether a player actually follows the prescribed action is left to the node's discretion.
We assume that a player cannot deviate in its internal deterministic computations; thus, the only possible deviations arise from (1) internal computations involving randomness or (2) external actions such as sending incorrect messages, omitting required messages, outputting values at unintended times, or outputting values different from those prescribed by the algorithm. 
This assumption is without loss of generality, since any deviation in deterministic internal computation must eventually manifest in the node’s external behavior to have any observable effect. Therefore, it suffices to model deviations as occurring in the external actions of a node rather than in its internal deterministic computation.
\end{remark}

\paragraph*{Intuition for our utility function and modeling choices}
Our aim in this paper is to study rational behavior in the LOCAL model. A basic feature in many LOCAL model algorithms is that the nodes perform computation for a certain number of rounds and finally produce a correct output upon termination. When the nodes are rational agents, we want them to terminate at some point; therefore, it is essential that once a node commits to an output, it cannot be changed. Moreover, our utility function is well-defined only once a node and its neighbors have computed their outputs, and allowing nodes to change their outputs makes it difficult to calculate their utility. We can still design algorithms where nodes are allowed to change their output by having nodes compute a temporary output that can be changed. But before termination, the nodes have to irrevocably commit to this output, and the utility is calculated based on the committed outputs of all nodes.

Next, we motivate our choice of utility function. If all nodes are in a locally valid solution, we have computed a (globally) correct MIS, and since we are interested in computing a correct MIS solution, we require that all nodes prefer a locally valid solution over a locally invalid solution. In fact, we do not want the nodes to try to join the MIS at a risk of being in a locally invalid solution with some probability. Hence, we assign a utility of \(-\infty\) to nodes in a locally invalid solution. Moreover, MIS has the property that if a node is in a locally invalid solution, it can always change its own output to be in a locally valid solution. Since an output change is not allowed, this means that a node in a locally invalid solution has committed too early to its output, and such behavior must be appropriately penalized.

If we have the same utility for all locally valid solutions, we are in the standard LOCAL model, as the rational nodes will faithfully follow any algorithm that computes a correct MIS. Therefore, to make the model interesting, we need to assign different utility values to different locally valid solutions. There are indeed many choices here, but for the purposes of this paper, we study one family of utility functions where nodes always prefer locally valid solutions where they belong to the MIS. The other choices in this space are also interesting, especially the case where nodes prefer locally valid solutions where they are not in the MIS. Designing rationality-resilient algorithms for these alternative families of utility functions is left as an open problem.

We now address why we introduce the third output symbol, \(\bot\). If we define the utility function with only two output labels \(0\) and \(1\), we may encounter the following scenario: if a node outputs $1$ by deviating from the protocol to get higher utility, then the best response for all its neighbors is to output $0$, since otherwise they obtain $-\infty$ utility. Therefore, the node that outputs \(1\) is getting away with its deviation simply because its neighbors can do nothing but support its action. 
To address this problem, we need a way for nodes to retaliate when their neighbors deviate from the algorithm and irrevocably output \(1\). We do this by introducing a new output symbol \(\bot\) which can be interpreted as a credible threat. Now, if a node \(i\) deviates by outputting \(1\), its neighbors can output \(\bot\) instead of \(0\) at no extra penalty to themselves, which makes the deviation of \(i\) not profitable. It is important to note that the output \(\bot\) is only used as a credible threat to ensure nodes do not deviate from the algorithm. In a correct execution of the algorithm with no deviations, no node will ever actually output \(\bot\). So a node outputting \(\bot\) can be considered as a failure of the algorithm execution.

We also assume that the output of each node is visible to all of its neighbors.   This assumption is necessary because, in our model, the utilities of nodes are based on their own output and the outputs of their neighbors. If outputs of neighbors are not directly observable, a node would have to rely solely on what its neighbors report as their output, which introduces the possibility of misreporting. In such a case, a trivial but undesirable outcome could arise where all nodes output \(1\) while falsely reporting to have output \(0\) to their neighbors, resulting in a configuration that is not a valid MIS, as every node is part of the MIS in reality. Hence, visibility of the neighbor's output is required to ensure that utilities are well-defined and correspond to actual local configurations.

\subsection{MIS with Rational Nodes as Extensive-Form Game}\label{subsec:ext} 
As discussed in Section~\ref{sec:model}, we consider a model in which nodes are \emph{rational} and interact with one another to change their states. 
Each node’s utility depends on the final state of the network, making the overall system naturally analogous to a strategic game among the nodes. Since these notions are well formalized in the game-theoretic framework, we model our setting as an \emph{extensive-form game}.

In an \(r\)-round LOCAL model algorithm, a node has no information about the graph outside its \(r\)-hop neighborhood. This means the corresponding extensive form game must be one with \emph{imperfect information}. Since each node retains complete memory of all its past actions and observations, the game satisfies the property of \emph{perfect recall}. Finally, because we assume synchronous communication rounds in which all nodes act simultaneously, the game can be viewed as an \emph{extensive-form game with imperfect information, perfect recall, and simultaneous moves}.

Formally, an extensive-form game is defined as a tuple consisting of several components 
(see \cite{gametheorybook} for a complete description). 
In our context, we do not require the full generality of this definition; 
it suffices to specify only the elements necessary to model our distributed setting. 
Accordingly, we work with a simplified version of the tuple and its associated definitions.  

\begin{definition}[Distributed MIS Problem with Rational Agents as an Extensive Form Game with Imperfect Information, Perfect Recall and Simultaneous Moves]\footnote{This is an instance of an extensive form game, refer \cite[Chapter 12]{gametheorybook} for the exact definitions of extensive form game
with imperfect information, perfect recall and simultaneous moves.}\label{def: MIS as game}
Given a graph \( G = (V, E) \), where \(V\) is the set of nodes (agents) and \(E \subseteq V \times V\) is the set of edges and each edge \((i, j) \in E\) indicates that agents \(i\) and \(j\) can communicate directly.  The MIS problem over \(G\) can be modeled as an extensive game, denoted by \(\mathcal{G}\). Formally, we define \(\mathcal{G}\) as a tuple 
\(
\mathcal{G} = \langle V, H, P,  \{\mathcal{I}_i\}_{i \in V}, \{A_i\}_{i \in V}, \{\succsim_i\}_{i \in V} \rangle
\)
whose components are described below.
\begin{itemize}
  \item \(V\): the set of players, corresponding to the nodes of the graph.

\item \(H\): the set of all possible prefix-closed histories of messages exchanged between neighboring nodes 
and their outputs, subject to the following constraints: each node produces an output at most once, 
and once a node outputs, it becomes inactive—i.e., it can no longer send or receive messages. 

  \item A history \( h \in H \) is called \emph{terminal} if it is infinite, or if all nodes output a value from \(\{0,1,\bot\}\). Denote by \(Z\) the set of all terminal histories.

\item \(P : H \setminus Z \to 2^V \setminus \{\emptyset\}\): assigns to each nonterminal history \(h\) the set of players who move simultaneously at that history. P(h) is the set of nodes that have not yet produced an output at history \(h\)

\item Since any player \(i\in V\) observes only the messages it has sent and received, together with the outputs of its neighbors, so every history \(h\in H\) has a local projection onto the information accessible to \( i \). Formally, we define the projection function
\(
\mathrm{proj}_i : H \to H_i,
\),
where \( H_i \) denotes the set of all possible local histories of player \( i \). For any global history \( h \in H \), the projection \( \mathrm{proj}_i(h) \) is the subsequence of events in \( h \) that involve \( i \)—namely, all messages sent or received by \( i \), together with the outputs of its neighbors observed by \( i \). 
Two global histories \( h, h' \in H \) are said to be \emph{indistinguishable} to \( i \) if they induce the same local history, that is,
\(
h \sim_i h' \iff \mathrm{proj}_i(h) = \mathrm{proj}_i(h').
\)
This indistinguishability relation \( \sim_i \) partitions the set \( \{ h \in H \mid i \in P(h) \} \) 
of global histories where \( i \) is active into equivalence classes.
We denote this partition by \( \mathcal{I}_i \), where each element \( I_{i} \in \mathcal{I}_i \) corresponds to the set of all global histories that induce the same local history to \(i\). Each such class \(I_i\) is therefore an information set of player \(i\), since
player \(i\) cannot distinguish between any two global histories in the same class
and therefore must choose the same action at all of them. Let \(I_i(h_i)\) denote the set of all global histories that project local history \(h_i\) on \(i\). 
\item At any information set \(I_i\in \mathcal{I}_i\).  
The set of possible actions for player \(i\),  denoted \(A_i(I_i)\), consists of any combination of:  
(1) perform internal computation,  
(2) send arbitrarily long messages to any subset of undecided neighbors, and  
(3) output a value in \(\{0,1,\bot\}\).

\item For each player \( i \in V \), a preference relation \( \succsim_i \) is defined over terminal histories \(Z\), representable by a local utility function \( u_i: Z \to \mathbb{R} \). For an infinite history \( h \in Z \), the utility of all nodes is defined to be \(0\). For a finite terminal history, the utility is given by the function in Equation~\ref{eq:utility}.

\end{itemize}
\end{definition}
Given the definition of the game, we can define pure, mixed, and behavioral strategies as follows:
\begin{definition}[Pure, Mixed and Behavioral Strategies {\cite[Definitions~203.1, 212.1]{Osborne1994}}]
A \emph{pure strategy} of player \(i \in V\) of an extensive-form game \(
\mathcal{G}\) is a function that assigns 
an action in \(A_i(I_i)\) to each information set \(I_i \in \mathcal{I}_i\).
A \emph{mixed strategy} of player \(i\) is a probability measure over the set of
\(i\)’s pure strategies.  Whereas a \emph{behavioral strategy} of player \(i\) is a
collection \((\beta_i(I_i))_{I_i \in \mathcal{I}_i}\) of independent probability measures, where \(\beta_i(I_i)\)
is a probability measure over \(A(I_i)\).

\end{definition}

\begin{remark}
Under this formulation, a strategy algorithm can be viewed as a \emph{behavioral strategy}, 
since it specifies, for each node, the action (or distribution over actions) it should take at every possible state it may encounter, which corresponds to the local history in the extensive form game \(\mathcal{G}\) defined above. 
\end{remark}

Next, we define various equilibrium notions for extensive-form games before introducing the equilibrium guarantee that can be obtained.

\begin{definition}[Nash equilibrium {\cite[Chapter 11.5]{Osborne1994}}]
A Nash equilibrium in mixed strategies of an extensive game is
a profile \(\sigma^*=(\sigma^*_i)_{i\in N}\) of mixed strategies with the property that for every 
player \(i \in N\) we have
\(
O(\sigma^*_{-i}, \sigma^*_i) \succeq_i O(\sigma^*_{-i}, \sigma_i)
\)
for every mixed strategy \(\sigma_i\) of player \(i\),
where \(O(\sigma)\) denotes the expected outcome (or distribution over terminal histories) 
induced by the strategy profile \(\sigma\) and \( \sigma^*_{-i} \) denotes the equilibrium strategy profile of all players 
except \( i \).
\end{definition}

\begin{definition}[Sequential Equilibrium{\cite[Definition 222.1]{Osborne1994}}]
A \emph{sequential equilibrium} of an extensive-form game 
\(\mathcal{G}\)
is an \emph{assessment} 
\((\beta, \mu)\),
where 
\(\beta = (\beta_i)_{i \in N}\) is a profile of \emph{behavioral strategies} 
(which assign, for each information set \(I_i \in \mathcal{I}_i\), a probability distribution over the available actions at \(I_i\)) 
and \(\mu\) is a \emph{belief system}, that is, a function that assigns to every player \(i\) 
and every information set \(I_i \in \mathcal{I}_i\) 
a probability measure over the set of histories contained in \(I_i\), conditional
on \(I_i\) being reached.
The assessment \((\beta, \mu)\) is a sequential equilibrium if it satisfies the following conditions:
\begin{enumerate}
    \item \textbf{Sequential Rationality:} 
    An assessment \((\beta, \mu)\) is said to be \emph{sequentially rational} if, 
    for every player \(i \in V\) and every information set \(I_i \in \mathcal{I}_i\), 
    the strategy \(\beta_i\) prescribes actions that maximize \(i\)’s expected utility at \(I_i\),
    given the belief system \(\mu\) and the strategies \(\beta_{-i}\) of the other players.
    Formally, let \(O(\beta, \mu \mid I_i)\) denote the distribution over terminal histories
    induced by \((\beta, \mu)\) conditional on \(I_i\) being reached.
    Then \((\beta, \mu)\) is sequentially rational if
    \(
    O(\beta, \mu \mid I_i)
    ~\succeq_i~
    O\bigl((\beta_{-i}, \beta^{'}_i), \mu \mid I_i\bigr)
    \quad \text{for every behavioral strategy } \beta^{'}_i \text{ of player } i.
    \) 
    \item \textbf{Consistency:} An assessment \((\beta, \mu)\) is said to be \emph{consistent} if there exists a sequence of assessments \(\{(\beta^n, \mu^n)\}_{n=1}^\infty\) that converges to \((\beta, \mu)\) in Euclidean space, where each strategy profile \(\beta^n\) is completely mixed and each belief system \(\mu^n\) is derived from \(\beta^n\) using Bayes' rule.

\end{enumerate}
\end{definition}

Another equilibrium notion that is considered is called \emph{Trembling Hand Perfect Equilibrium}, which is a refinement of sequential equilibrium with a particular assessment. It ensures that each player's strategy remains optimal, even if players occasionally make unintended moves with a very small probability, thereby ensuring that strategies are robust to such "trembles" in decision-making. While both sequential and trembling hand perfect equilibria provide a guarantee only with respect to a particular belief system (that is, a strategy is optimal given that belief system), 
Our strategy algorithms give behavioral strategies that admit a stronger guarantee: they are Sequentially Rational for \emph{any} belief system. 
Hence, unlike sequential equilibrium, we do not define behavioral strategies relative to a belief system, 
But only over information sets. 
We formalize this stronger notion as a \EQ.
\begin{definition}[\EQ]\label{def:EQ}
A behavioral strategy profile 
\(\beta^\ast = (\beta_i^\ast)_{i \in N}\)
of an extensive-form game \(\mathcal{G}\) 
is an \emph{\EQ} if, for every player \(i\) 
and every information set \(I_i \in \mathcal{I}_i\), 
the behavioral strategy \(\beta_i^\ast\) is sequentially rational, 
regardless of player \(i\)’s belief over the histories in \(I_i\).
Formally, for every player \( i \in V \), every information set \( I_i \in \mathcal{I}_i \),  
any belief system \( \mu \) of player \( i \),
\(
O(\beta^\ast, \mu \mid I_i)
~\succeq_i~
O\bigl((\beta_{-i}^\ast, \beta'_i), \mu \mid I_i\bigr)
\) for every behavioral strategy \( \beta^{'}_i\) of player \(i\). 
where \(\beta_{-i}^\ast\) denotes the equilibrium strategies of all players other than \(i\) 
in the profile \(\beta^\ast\).
\end{definition}

Now we are ready to define the desirable properties of a strategy algorithm.
\begin{definition}[Rationality Resilient Strategy Algorithm]\label{def:RR}
A \emph{Strategy Algorithm} is said to be \emph{rationality resilient}(with respect to the above MIS problem) if it satisfies the following properties 
\begin{enumerate}
  \item \textbf{\EQ:} The behavioral strategy corresponding to the strategy algorithm should induce \EQ{} \ref{def:EQ}.
    \item \textbf{Termination:} If all nodes follow the algorithm, it must terminate after a finite number of rounds with high probability.
    
\item \textbf{Correctness:} If all nodes follow the algorithm and terminate, the set of nodes that output \(1\) forms a valid Maximal Independent Set (MIS), and all remaining nodes output \(0\).

\item \textbf{Positive Inclusion Probability:} If all nodes follow the algorithm, then each node must have a non-zero probability of eventually outputting \(1\) in \emph{every round} where none of its neighbors has output \(1\). This ensures that no node prematurely drops out of participation based on the belief that it has no future chance of joining the MIS.
\end{enumerate}

\end{definition}

\subsection{Related Work}
\label{sec:related-work}
The study of incentive-compatible distributed algorithms has led to the field of \emph{Distributed Algorithmic Mechanism Design (DAMD)}~\cite{DAMD}, which extends the classical \emph{Algorithmic Mechanism Design (AMD)} framework to distributed settings where there is no trusted central authority. While AMD focuses on designing truthful mechanisms in centralized settings, DAMD addresses the challenge of ensuring that rational agents follow prescribed strategies when the computation is decentralized.

Nisan~\cite{Nisan99} gave a distributed mechanism for computing a truthful maximum independent set in chain networks. Later, Hirvonen and Ranjbaran~\cite{DLDM} proposed a distributed mechanism for all local optimization problems, including Maximal Independent Set, that guarantees a $\Delta$-approximation. However, both approaches require monetary transfers (payments) to third parties to enforce truthfulness -- something we explicitly aim to avoid in our work. Moving beyond the third-party setting, other widely studied notions include \emph{self-stabilization}\cite{SS}\cite{SSI} and \emph{fairness}\cite{DBLP:conf/ipps/FinemanNSW14}\cite{AbrahamDH2019}. The work by Amoussou-Gueno et al. \cite{RBB} analyzes the dynamics between rational and Byzantine players in blockchain consensus protocols, highlighting their implications for security and efficiency in decentralized systems.

Rationality has also been extensively studied in the context of the \emph{Secret Sharing} problem. 
Halpern et al.\ \cite{rssmp} introduced the notion of rational secret sharing, showing that classical schemes fail when agents act in a selfish manner, and proposed randomized protocols that restore incentive compatibility. Subsequent works~\cite{DMG,rss,rsscpr} explored alternative utility models and provided further constructions under these variations. 
Although the problem domain is different, these works share with ours the high-level goal of designing distributed algorithms that are resilient to rational deviations.

Solving \emph{Locally Checkable Labeling (LCL) problems} with rational nodes was formalized as extensive-form games by Collet et al.\ \cite{LDASA}. They show that for all LCL problems solvable by greedy sequential algorithms -- such as MIS, \((\Delta + 1)\)-coloring, and maximal matching -- there exist distributed algorithms that are robust to selfish behavior and form trembling hand perfect equilibria. While we adopt an equivalent extensive-form game framework, our model differs in several key respects. First, in their setting, actions are restricted to choosing labels, whereas in our model, nodes may also send arbitrary messages to their neighbors. Second, their formulation allows labels to be revised in later stages, while in our setting, label choices, once made, cannot be changed. Finally, in their model, the game always terminates in a valid LCL solution, whereas in our framework, the outcome may correspond to an invalid solution. In our work, we study the MIS problem under a specific family of utility functions and present a strategy that is sequentially rational irrespective of beliefs, which is a much stronger guarantee than trembling hand perfect equilibrium.

Fineman et al. \cite{DBLP:conf/ipps/FinemanNSW14} define fairness in the context of computing maximal independent sets as balanced inclusion probabilities across nodes.
However, they do not allow nodes to have rational behavior, and hence their algorithm is not robust to strategic deviations. 
In contrast, Abraham et al.\ \cite{AbrahamDH2019} study fair leader election for specific topologies such as chains and cliques. 
Their model assumes that agents strictly prefer the existence of some leader over no leader, which is conceptually similar to our assumption that agents prefer to be included in the MIS. 
Importantly, their protocols are also resistant to collusion. 
Our algorithms adopt a similar idea of generating randomness from inputs provided by multiple agents with conflicting preferences over the outcome.

\section{A Strategy Algorithm for MIS with Rational Nodes}
\label{sec:mis-rps}

We now describe our first strategy algorithm, which employs a simple and symmetric tie-breaking mechanism: each node plays a Rock-Paper-Scissors (RPS) game with each of its neighbors, and only the nodes that win with all their neighbors are allowed to join the MIS. We note that the RPS game between two nodes \(i\) and \(j\) must be well defined 
even if one or both players choose not to play the game. 
If both nodes send valid moves to each other, then the outcome is determined by 
the standard RPS rules: rock breaks scissors, scissors cut paper, paper covers rock. 
If exactly one of the two nodes sends an invalid move or sends no move at all, 
then the other player wins. 
If both nodes fail to send valid moves, the result is a tie.  We capture all these cases formally in the following definition.
\begin{definition}[RPS Outcome Function]\label{rps_function}
Let 
\(
\text{rps}_{i,j} : \{r, p, s, \phi\}^2 \to \{i, j, \text{tie}\}
\)
be a function which determines the winner between two nodes \(i\) and \(j\) that are playing an RPS game with each other. Here \(r\), \(p\), and \(s\) correspond to \emph{rock}, \emph{paper}, and \emph{scissors}, respectively, and \(\phi\) denotes an invalid or missing move (i.e., any input not in \(\{r, p, s\}\)). When node \(i\) plays move \(s_i\) and node \(j\) plays move \(s_j\), the function is defined as follows:

\[
\text{rps}_{i,j}\bigl(s_i, s_j\bigr) =
\begin{cases}
i, & \mathrm{if~} (s_i, s_j) \in \{(r,s), (p,r), (s,p), (r,\phi), (p,\phi), (s,\phi) \} \\[4pt]
j, & \mathrm{if~} (s_i, s_j) \in \{(s,r), (r,p), (p,s), (\phi, r), (\phi, p), (\phi, s)\} \\[4pt]
\text{tie}, & \mathrm{if~} s_i = s_j
\end{cases}
\]
\end{definition}

\paragraph{RPS Algorithm Overview}
The algorithm proceeds in \emph{iterations}, and each iteration consists of \emph{three rounds}. Initially, all nodes are in an \texttt{undecided} state. Once a node outputs a value in $\{0, 1, \bot\}$, it becomes \texttt{decided} and takes no further action.
Each node maintains a \texttt{BeenCheated} flag, initialized to \texttt{false}, and in every iteration, the strategy algorithm proceeds through the following rounds:
\begin{itemize}
    \item \textbf{Round 1 (Play RPS):} Each node \(i\) selects a move from \(\{r, p, s\}\) uniformly at random for each of its undecided neighbors and sends it. The RPS games with different neighbors are treated independently, allowing the node to play distinct moves against each of its neighbors.
 \item \textbf{Round 2 (Join MIS if possible):} Each node \(i\) checks whether any neighbors has already output \(1\) or \(\bot\). If such neighbors exist, they are considered to have deviated from the algorithm, since they were not supposed to do so under the prescribed algorithm. In that case, \(i\) ignores the rest of the algorithm and outputs \(\bot\) in every subsequent round where it remains \texttt{undecided}. Otherwise, if no neighbor has output \(1\) or \(\bot\), then if \(i\) wins against all its \texttt{undecided} neighbors according to Definition~\ref{rps_function}, it outputs \(1\) and joins the MIS. In addition, if \(i\) did not receive a move from an \texttt{undecided} neighbor in the previous round, it sets its \texttt{BeenCheated} flag to \texttt{true}.

 \item \textbf{Round 3 (Respond to neighbor outputs):} Each node \(i\) checks whether any neighbor \(j\) has output \(\bot\). 
If such a node exists, \(i\) outputs \(\bot\) in every subsequent round in which it remains \texttt{Undecided}. 

  If any neighbor \(j\) has output \(1\) in previous round, 
then:
\begin{itemize}
    \item If \(i\) lost to all such neighbors
    then node \(i\) outputs \(0\) in every subsequent round while it remains undecided, 
    unless \texttt{BeenCheated} is \texttt{true}, in which case it outputs \(\bot\) in every subsequent round while it remains undecided.
    \item Otherwise, if there exists a neighbor \(k\) that lost to \(i\) but still output \(1\) 
    then \(i\) outputs \(\bot\) in every subsequent round while it remains undecided.
\end{itemize}
Node \(i\) outputs \(1\) if it has no undecided neighbors; otherwise, it proceeds to the next iteration.

\end{itemize}

\subsection{Rationality Resilient RPS Algorithm }
\label{sec:rps-dev}
The distributed strategy algorithm described above is formalized in Algorithm \ref{alg:rps_dev}. This is followed by an analysis showing that this strategy is rationality resilient, encapsulated in Theorem \ref{thm:RPS_main_theorem}.

\begin{algorithm}
\caption{\label{alg:rps_dev}RPS based MIS Strategy algorithm for node \( i \)}
\SetKwBlock{Init}{Initialization}{}
\Init{
  \texttt{UndecidedNeighbors} \( \gets N(i) \)\;
    \texttt{BeenCheated} \( \gets \) false\;
}
\If{\(N(i) = \emptyset\)}{
From this point onward, disregard the rest of the algorithm and Output \(1\) in every round where node \(i\) is still \textsc{Undecided}\;
}

\While{\(i\) is \emph{\textsc{Undecided}}}{
    \tcc{Round 1: Play RPS with Undecided Neighbors}
    \ForEach{\( j \in \) \emph{\texttt{UndecidedNeighbors}}}{
        Choose strategy \( s_i(j) \in \{\text{r}, \text{p}, \text{s}\} \) uniformly at random and send to \( j \)\;
    }
    
    \tcc{Round 2: Join MIS if possible}
    \ForEach{\( j \in \) \emph{\texttt{UndecidedNeighbors}}}{
        \If{\( j \) has output \( 1 \) or \( \bot \) }{
       From this point onward, disregard the rest of the algorithm and Output \(\bot\) in every round where node \(i\) is still \textsc{Undecided}\;
        }
        \ElseIf{\( j \) has output \( 0 \)}{
            Remove \( j \) from \texttt{UndecidedNeighbors}\;
        }
    }
    \ForEach{\( j \in\) \emph{\texttt{UndecidedNeighbors}}}{
     \If{message \( s_j(i) \) is \(\phi \) }{
            \texttt{BeenCheated} \( \gets \) true\; 
        }
    Compute \( rps_{i, j}(s_i(j),s_j(i))\) as defined in Definition~\ref{rps_function};

    }
    \If{\( rps_{i, j}(s_i(j),s_j(i)) = i \) for all \( j \in \texttt{UndecidedNeighbors} \)}{
        Output \( 1 \) (join MIS)\;
    }

    \tcc{Round 3: Observe Undecided Neighbors' Outputs}
    \If{there exits a node in \emph{\texttt{UndecidedNeighbors}} has output \( \bot \)}{
    From this point onward, disregard the rest of the algorithm and output \(\bot\)
    }
    \ElseIf{there exits a node in \emph{\texttt{UndecidedNeighbors}} has output \( 1 \)}{
        \If{there exists a neighbor \( j \) with output 1 and \( rps_{i, j}(s_i(j),s_j(i)) \neq j \)}{
            \texttt{BeenCheated} \(\gets\) true\;
        }
        From this point onward, disregard the rest of the algorithm and output \( \bot \) if \texttt{BeenCheated} is \textbf{true}, and output \(0\) otherwise, in every round where node \(i\) is still \textsc{Undecided}\;
    }
    \Else{
        Remove all \( j \in \) \texttt{UndecidedNeighbors} who output \( 0 \)\;
        \If{\emph{\texttt{UndecidedNeighbors}} is empty}{
        From this point onward, disregard the rest of the algorithm and Output \(1\) in every round where node \(i\) is still \textsc{Undecided}\;
        }
    }
}
  \end{algorithm}

\begin{theorem}\label{thm:RPS_main_theorem}
Algorithm~\ref{alg:rps_dev} is a 
Rationality Resilient Algorithm as defined in Definition \ref{def:RR}.
\end{theorem}
We prove Theorem~\ref{thm:RPS_main_theorem} in four parts (Theorems~\ref{thm:rps_eq}, \ref{thm:rps_termination}, \ref{thm:rps_correctness}, and \ref{thm:rps_ppi}). Among these, Theorem~\ref{thm:rps_eq} requires intermediate results, which we establish through Lemmas~\ref{lem:rps_been_cheated} and~\ref{lem:rps_stick_to_algo}.

\begin{lemma}\label{lem:rps_been_cheated}
Let the game be at an arbitrary history $h$ with $i\in P(h)$, where $h \in I_i \in \mathcal{I}_i$. 
Suppose that one of $i$'s \texttt{UndecidedNeighbors} (say, $j$) has its \texttt{BeenCheated} flag set to \texttt{true}.  
Then, regardless of the history $h$ and of $i$'s strategy from $h$ onward, if all the other nodes follow Algorithm~\ref{alg:rps_dev} from $h$ onward, the utility of node $i$ will be at most~$0$.
\end{lemma}

\begin{proof}
    According to Algorithm \ref{alg:rps_dev}, node \( j \) outputs \( 0 \) only in Round 3 of an iteration, and even then, only if its \texttt{BeenCheated} flag is \texttt{false}. Since the \texttt{BeenCheated} flag, once set to \texttt{true}, is never reset to \texttt{false}, a node with \texttt{BeenCheated} = \texttt{true} will never output \( 0 \). Therefore, node \( j \) will only output \( 1 \) or \( \bot \). In either case, the utility of all its neighbors, including \(i\), is at most \( 0 \). 
\end{proof}

\begin{lemma}\label{lem:rps_stick_to_algo}
    Let the game be at an arbitrary history $h$ with $i\in P(h)$, where $h \in I_i \in \mathcal{I}_i$, and let \( h_i \) be the local history of node \( i \). 
    Then, provided that \( i \) and all other undecided nodes follow Algorithm~\ref{alg:rps_dev} from this point onward, the expected utility of node \( i \) is at least \( 0 \).
\end{lemma}

\begin{proof}
    A node \( i \) receives a utility less than zero (i.e., \( -\infty \)) only in two cases:
    \begin{enumerate}
        \item \( i \) outputs \( 0 \), and none of its neighbors output \( 1 \), or
        \item \( i \) outputs \( 1 \), and at least one of its neighbors also outputs \( 1 \).
    \end{enumerate}

    \textbf{Case 1} is not possible under the Algorithm~\ref{alg:rps_dev}: a node is instructed to output \( 0 \) only in Round 3 of an iteration, and only after observing at least one neighbor output \( 1 \). Hence, after any local history where the node \(i\) is undecided in Round 3 outputs \( 0 \), it must be the case that a neighbor has already output \( 1 \). Therefore, the first case cannot occur.

  \textbf{Case 2} is also prevented by Algorithm~\ref{alg:rps_dev}: a node outputs \( 1 \) only in Round 2 or Round 3 of an iteration.  
\begin{itemize}
    \item In Round 2, node \( i \) outputs \( 1 \) only if no neighbor has yet output \( 1 \) or \( \bot \), and \( i \) has won the RPS game against all its undecided neighbors. Since all neighbors also follow the algorithm, any neighbor that loses to \( i \) in the RPS game will not output \( 1 \). Thus, no neighbor of \( i \) can simultaneously output \( 1 \).
    \item In Round 3, node \( i \) outputs \( 1 \) only if all of its neighbors have already output \( 0 \). Hence, it is impossible for a neighbor to output \( 1 \) in this case.
\end{itemize}
Therefore, Case 2 cannot occur. Combined with the fact that Case 1 is also impossible, we conclude that by following Algorithm~\ref{alg:rps_dev}, node \( i \)'s utility is never less than \( 0 \).
\end{proof}

\begin{theorem}\label{thm:rps_eq}
The strategy corresponding to Algorithm \ref{alg:rps_dev} constitutes a \EQ.
\end{theorem}

\begin{proof}
Consider an arbitrary history $h$ with $i \in P(h)$. Let $h_i$ denote the local history of node $i$, which implies $h \in I_i(h_i) \in \mathcal{I}_i$. We show that the continuation strategy prescribed for $i$ by Algorithm~\ref{alg:rps_dev} from information set $I_i(h_i)$ is sequentially rational irrespective of the \(i\)'s belief over $I_i(h_i)$, i.e., it gives the maximum expected payoff when all other nodes also follow Algorithm~\ref{alg:rps_dev} from $h$. It follows that the strategy constitutes a sequential equilibrium for any belief system of \(i\).
 and it follows that the strategy constitutes a 
\EQ.
As we mentioned in Section~\ref{sec:model}, we do not consider internal deviations except for randomness generation.
We distinguish deviations as:
\begin{itemize}
\item \emph{Undetectable deviations:} for example, choosing RPS moves with nonuniform
probabilities. Since every sequence of moves is still possible under the prescribed
uniform strategy (albeit with small probability), neighbors can never rule out a deviation
with certainty, even after many observations. Similarly, voluntarily outputting \(0\) or
\(\bot\) earlier than prescribed may also be indistinguishable, since it could be a valid
response to another neighbor having output \(1\), which remaining neighbors may not
observe.
  \item \emph{Detectable deviations:} e.g., failing to send a required/valid message,
  sending inconsistent messages, or outputting \(1\) when not eligible. Such deviations
  are caught by at least one neighbor, upon which the neighbor sets its \texttt{BeenCheated} flag to \(1\).
\end{itemize}

\noindent
\textbf{Case analysis.}

\smallskip
\emph{Case A: At \(h\), \(i\) has a decided neighbor with output \(1\) or \(\bot\).}
Then the maximum attainable payoff for \(i\) is \(0\). By Lemma~\ref{lem:rps_stick_to_algo},
following the algorithm yields a payoff of at least \(0\). Hence, no profitable deviation exists,
so the prescribed action is a best response.

\smallskip
\emph{Case B: At \(h\), \(i\) do not have a decided neighbor with output \(1\) or \(\bot\) and there exists an
undecided neighbor that \(i\) has cheated.}
If the remaining undecided neighbors follow the algorithm, Lemma~\ref{lem:rps_been_cheated}
implies the most \(i\) can obtain is \(0\); following the algorithm also yields payoff \(0\).
Thus, no profitable deviation exists.

\smallskip
\emph{Case C: At \(h\), \(i\) does not have a decided neighbor with output \(1\) or \(\bot\) and \(i\) has not cheated
any undecided neighbor.}
\begin{itemize}
  \item \emph{Detectable deviation by \(i\):} Some neighbor detects it and sets
  \texttt{BeenCheated}\(=1\). By Lemma~\ref{lem:rps_been_cheated}, the best payoff \(i\) can then secure is \(0\).
  By Lemma~\ref{lem:rps_stick_to_algo}, sticking to the algorithm yields a payoff of at least \(0\),
  so no profitable deviation exists.
  \item \emph{Undetectable deviation by \(i\):} (i) Changing the probabilities with which \(i\) plays rock, paper, scissors does not improve \(i\)'s win probability against neighbors who pick each move uniformly at random. Therefore, \(i\)'s chance
  of joining the MIS is unchanged, and so is the expected payoff. (ii) Voluntarily outputting
  \(0\) or \(\bot\) can yield a payoff at most \(0\), whereas Lemma~\ref{lem:rps_stick_to_algo} guarantees that following the algorithm yields payoff at least \(0\). Hence, no strictly profitable deviation exists.
\end{itemize}

Therefore, we have shown that for every history \(h\) with \(i\in P(h)\),
the prescribed strategy for \(i\) based on \(h_i = \mathrm{proj}_i(h)\) gives maximum expected utility possible from \(h\) when all other node are following the algorithm from \(h\), which proves the theorem.
\end{proof}

We have shown that Algorithm \ref{alg:rps_dev} forms a \EQ. To complete the proof of Theorem~\ref{thm:RPS_main_theorem}, we now assume that all nodes follow the algorithm without deviation and, under this honest execution, we analyze the simplified version of the algorithm (with deviation-handling logic removed) and proceed to prove its correctness, termination, and nonzero inclusion probability.

\subsection{Strategy Algorithm under No Deviations}
\label{sec:rps-nodev}
Since we have established in Theorem~\ref {thm:rps_eq} that Strategy Algorithm~\ref{alg:rps_dev} constitutes a Belief-Independent Sequential Equilibrium, no node has a unilateral incentive to deviate. When no node deviates from the Strategy Algorithm~\ref{alg:rps_dev}, its execution is equivalent to running Algorithm~\ref{alg:rps_nodev}, which is simply Strategy Algorithm~\ref{alg:rps_dev} with the deviation handling logic removed.

\begin{algorithm}
\caption{\label{alg:rps_nodev} RPS based MIS strategy algorithm for node \(i\) without deviations}
\SetKwBlock{Init}{Initialization}{}
\Init{
    Initialize node \( i \) as undecided\;

    \texttt{UndecidedNeighbors} \( \gets N(i) \)\;
    }

\While{true}{
    \tcc{Round 1: Play RPS with Undecided Neighbors}
    \ForEach{\( j \in \texttt{UndecidedNeighbors} \)}{
        Choose strategy \( s_i(j) \in \{\text{r}, \text{p}, \text{s}\} \) uniformly at random and send to \( j \)\;
    }
    \tcc{Round 2: Join MIS if possible}
    \If{\( rps_{i, j}(s_i(j),s_j(i)) = i \) for all \( j \in \texttt{UndecidedNeighbors} \)}{
        Output \( 1 \) (join MIS and terminate)\;
    }
    \tcc{Round 3: Observe Undecided Neighbors' Outputs}
    \If{any \( j \in \texttt{UndecidedNeighbors} \) has output \( 1 \)}{
                Output \( 0 \) and terminate\;
        }
     Remove all \( j \in \texttt{UndecidedNeighbors} \) who output \( 0 \)\;
     \If{\emph{\texttt{UndecidedNeighbors}} is empty}{
        Output \(1\) (join MIS and terminate)\;
        }
     }

\end{algorithm}

We now proceed to prove the remaining claims of Theorem~\ref{thm:RPS_main_theorem} on this deviation-free execution: namely, correctness (that the output forms a valid MIS), termination (with high probability), and positive inclusion probability (every node has a nonzero chance of joining the MIS).

\begin{theorem}[Termination]\label{thm:rps_termination} Algorithm \ref{alg:rps_nodev} terminates in \(O(3^{4\Delta}\log n)\) rounds with high probability on graphs with maximum degree \(\Delta\).    
\end{theorem}
\begin{proof}
    Let \( Y \) denote the number of edges removed in a given iteration, and for each edge \( \{i,j\} \in E \), define $X_{ij}$ to be an indicator variable for the event that $i$ wins with all its neighbors and $j$ wins with all its neighbors except for $i$. Let \(d_i\) and \(d_j\) be the number of undecided neighbors of \(i\) and \(j\) respectively. As $i$ needs to win a total of \(d_i\) games and $j$ needs to win a total of \(d_j-1\) games, we have,
\(
\Pr[X_{ij} = 1] = 1/(3^{d_i + d_j-1})
\)

We now relate this to the total number of edges that have been removed. Let \( Y \) be the total number of removed edges, and let both \( Y_{ij} \) and \( Y_{ji} \) denote the removal of edge \(\{i,j\}\). For all vertices \(i\in N(j)\) there can only be one node \(k\) with \(X_{kj}=1\), and if there is such a node \(k\), then \(Y_{ij}=1\) for all \(i\in N(j)\) (as \(k\) will join the MIS and cause \(j\) to be delete and hence all edges incident on \(j\) will be deleted).
Therefore, \(\sum_{i\in N(j)}d_jX_{ij}\leq\sum_{i\in N(j)}Y_{ij}\).
Now summing over all the vertices gives
\[\sum_{j\in V}\sum_{i\in N(j)}d_jX_{ij}\leq \sum_{j\in V}\sum_{i\in N(j)}Y_{ij} \implies\sum_{ij\in E}(d_jX_{ij}+d_iX_{ji})\leq \sum_{ij\in E}2Y_{ij}\]

Taking the expectation on both sides,
\[
2 \mathbb{E}[Y] \geq \sum_{ij \in E} \left( d_j \cdot \mathbb{E}[X_{ij}] + d_i \cdot \mathbb{E}[X_{ji}] \right)
 \geq \sum_{ij \in E} \left( \frac{3 d_j}{3^{d_i + d_j}} + \frac{3 d_i}{3^{d_j + d_i}} \right) = \sum_{ij \in E} \frac{3(d_i + d_j)}{3^{d_i + d_j}}
\]

Let \( \Delta \) be the maximum degree in the graph. Since \( d_i, d_j \leq \Delta \) and \( d_i + d_j \geq 2\), we get
\[
\mathbb{E}[Y] \geq \frac{3}{2} \sum_{ij \in E} \frac{d_i + d_j}{3^{2\Delta}} = \frac{3}{2 \cdot 3^{2\Delta}} \sum_{ij \in E} (d_i + d_j) \geq \frac{3}{2 \cdot 3^{2\Delta}} \cdot \sum_{ij \in E}2 = \frac{3|E|}{3^{2\Delta}}
\]

Let \( \alpha = {3}/{(2 \cdot 3^{2\Delta})}\). Therefore \(\mathbb{E}[Y] \geq 2c|E|\) . Using Lemma~\ref{Markov}, with \( d = \alpha|E| \) and \( a = |E| \):
\[
\Pr\left(Y \geq \alpha|E| \right) \geq \frac{\mathbb{E}[Y] - \alpha|E|}{|E| - \alpha|E|}
\geq \frac{\alpha}{1 - \alpha}
\ge \alpha
\]

Thus, with probability at least \(\alpha = {3}/{(2 \cdot 3^{2\Delta})} \), at least an \(\alpha\)-fraction of the edges are removed in a single iteration. Let \(E' \) be the set of edges remaining after the iteration, we have \( |E'| \leq (1 - \alpha) |E|\) with probability at least \(\alpha\). 

To reduce the number of edges from \(|E| \le n^2\) to 0, it suffices to have \( r \) successful iterations of edges reducing by factor of \((1 - \alpha)\) such that \( n^2 \cdot \left(1 - \alpha\right)^r \le n^2 \cdot e^{-\alpha r} < 1 \) which implies \(r > (2/\alpha)\ln n\). Now we will calculate the number of trials $T$ required to get $r$ successful iterations with high probability.

For each iteration \(1 \le t \le T\), let $Z_t$ be an indicator variable denoting the event that at least a $(1-\alpha)$-fraction of edges were removed in iteration $t$. Let \( Z = \sum_{t=1}^T Z_t\) be the total number of successful rounds out of \( T \) total rounds. By linearity of expectations, we have \( \mathbb{E}[Z] \geq \alpha T\). Now consider \( T = (4c/\alpha^2) \log n \) for some constant $c>0$.  
We have \(\mathbb{E}[Z] \geq (4c/\alpha) \log n\). Applying the Chernoff bound~\ref{chernoff} with $\delta=1/2$, we obtain
\[
\Pr\left(Z \leq (2c/\alpha) \log n \right) 
    \leq \exp\!\left(-\frac{4c \log n}{8\alpha} \right) \ll n^{-c}
\]

So, if we have $T$ trials, we get that whp, more than $(2c/\alpha) \log n$ successful iterations, which implies that all edges are removed and the algorithm terminates. Therefore, the algorithm terminates in $O(3^{4\Delta}\log n)$ rounds with high probability, which proves the lemma.
\end{proof}

\begin{theorem}[Correctness] \label{thm:rps_correctness}
When all nodes follow Algorithm~\ref{alg:rps_nodev}, after termination, the set of nodes that output \(1\) forms a valid Maximal Independent Set (MIS) of the network graph, and all remaining nodes output \(0\).
\end{theorem}

\begin{proof}
We now prove that the output of Algorithm \ref{alg:rps_nodev} satisfies both the independence and maximality conditions of a Maximal Independent Set (MIS).
\begin{itemize}
    \item \emph{Independence:} A node \(i\) outputs \(1\) (i.e., joins the MIS) only if it wins the Rock-Paper-Scissors (RPS) game against all its currently undecided neighbors, given that all its decided neighbors have output only \(0\). In no case, two neighboring nodes \( i \) and \( j \) both win against each other in the same iteration. Therefore, no two adjacent nodes can simultaneously output \(1\), and once a node 
outputs \(1\), its neighbors are required to output \(0\) in the next round, 
thereby ensuring independence.

\item \emph{Maximality:} A node \( i \) outputs \( 0 \) only in Round 3 and only if at least one of its neighbors has output \( 1 \). Moreover, each node continues participating until it outputs either \( 0 \), \( 1 \), or \( \bot \). However, under the assumption that all nodes follow Algorithm~\ref{alg:rps_nodev} honestly, no node ever outputs \( \bot \) (as no node deviates). Thus, since termination is guaranteed with high probability by Theorem~\ref{thm:rps_termination}, every node eventually outputs \( 1 \) or \( 0 \), and every node that outputs \( 0 \) must be adjacent to a node in the MIS. Hence, the resulting set is maximal.
\end{itemize}

\end{proof}

\begin{theorem}[Positive Probability of MIS Inclusion]\label{thm:rps_ppi}
Under Algorithm \ref{alg:rps_nodev}, each node has a non-zero probability of eventually outputting \(1\), i.e., joining the MIS, in every round where none of its neighbors has output \(1\) 
\end{theorem}
\begin{proof}
In any iteration where a node \(i\) is still undecided and none of its neighbors has output \(1\), it plays the RPS game with all of its undecided neighbors. Since the outcome of RPS is uniformly random, the probability that \(i\) wins against any particular neighbor is \(\tfrac{1}{3}\). Consequently, the probability that \(i\) wins against all of its undecided neighbors and joins the MIS in that iteration is \(\left(\tfrac{1}{3}\right)^{d_i}\), where \(d_i\) denotes the number of undecided neighbors of \(i\). Thus, as long as none of its neighbors has output \(1\), node \(i\) always has a strictly positive probability of joining the MIS.
\end{proof}

\section{A Faster Strategy Algorithm using Cryptography}
\label{sec:mis-crypto}
The RPS-based algorithm described in the previous section is not very efficient on graphs with high-degree nodes. The exponential dependence on degree arises because each node plays independent RPS games with its neighbors and can only join the MIS if it wins \emph{all} of them.
To improve the round complexity, we require an algorithm where each node plays a single game with all of its neighbors, under the condition that if a node wins, none of its neighbors can win. The algorithm by Métivier et al.~\cite{Metivier}, can be viewed as such a game: each node generates a random rank, and a node ``wins'' if its rank is lower than that of all its neighbors, in which case it joins the MIS. However, as discussed in Section~\ref{sec:challenges}, this algorithm by itself fails in rational settings because nodes can bias their rank selection to gain an advantage.

A natural attempt to get around this obstacle is to have each node's rank also depend on inputs from its neighbors. For instance, if rank of node \(i\) is computed using randomness from both \( i \) and an arbitrarily chosen neighbor \( j \), then neither \(i\) nor \(j\) can fully influence the outcome: \( i \) prefers a low rank in order to join the MIS, while \( j \) prefers \(i\) to have a high rank in order to block it from joining the MIS. This idea of \emph{two-party shared randomness} ensures that the ranks will be \emph{truly random} and cannot be changed unilaterally by any single party.
To make this verifiable, nodes must prove the correctness of their computed rank to others. This requires sharing the opponent's input, which must be authenticated by some means. We perform authentication using digital signatures, which requires standard cryptographic assumptions.

\paragraph*{Additional Assumptions} 
As mentioned, our approach relies on cryptographic tools; we assume that each node \(i\) is equipped with a public–private key pair and can sign messages. In order to verify the digital signatures, all nodes are assumed to know the identifiers and public keys of every other node in the network. Further, nodes can only broadcast to their neighbors.\footnote{The algorithm can, in principle, be extended to support unicast communication. However, doing so introduces several corner cases that would require substantially more discussion. One such case arises when a node can avoid detection by selecting multiple opponents simultaneously. Although this strategy does not yield a higher expected payoff than choosing a single opponent, we must still account for it, since the algorithm is required to prescribe the best continuation strategy for every possible history.}

\subsection{Rationality Resilient Rank-based Strategy Algorithm}
The algorithm proceeds in \emph{iterations}, where each iteration consists of \emph{five rounds}. Initially, all nodes are in the \texttt{undecided} state. Once a node outputs a value in \( \{0,1,\bot\} \), it becomes \texttt{decided}, meaning that its output is visible to all neighbors and it takes no further action in subsequent rounds. Each node maintains a \texttt{BeenCheated} flag, initialized to \texttt{false}, and in every iteration, the algorithm proceeds through the following rounds:
\begin{itemize}
    \item \textbf{Round 1: Opponent Selection.}
Each node \( i \) selects an arbitrary undecided neighbor as its \emph{opponent}, denoted by opp(i), and broadcasts this choice to all neighbors.
\item \textbf{Round 2 (Randomness Generation and Exchange).}
Each node \(i \) generates a \(c\log n\) length bit string \( r_{i \rightarrow i} \in \{0,1\}^{c\log n} \) uniformly at random and broadcasts it to all neighbors, where $c > 0$ is a fixed constant. For each neighbor \( j \) who selected \( i \) as an opponent in Round 1, node \( i \) generates a separate \(c\log n\) length bit string \( r_{i \rightarrow j} \in \{0,1\}^{c\log n} \) uniformly at random, signs it along with the current iteration counter, and broadcasts each signed message.

\item \textbf{Round 3: Share Randomness of Opponent} Each node \(i\) that broadcast \(\operatorname{opp}(i)\) in Round~1 and received \(r_{\operatorname{opp}(i) \rightarrow i}\) 
computes its rank as 
\(R_i = r_{i \rightarrow i} \oplus r_{\operatorname{opp}(i) \rightarrow i}\), 
and then broadcasts the signed message it received from \(\operatorname{opp}(i)\) to all its neighbors.

\item \textbf{Round 4: Calculate Ranks of Neighbors.} 
Each node \(i\), for each of its neighbors \(j\), after receiving a valid signed 
\(r_{\mathrm{opp}(j) \rightarrow j}\) in the previous round, 
computes the rank 
\(R_j = r_{j \rightarrow j} \oplus r_{\mathrm{opp}(j) \rightarrow j}\). If any undecided neighbor \(j\) did not broadcast some message in the previous 3 rounds as expected, 
then \(i\) assumes \(j\)’s rank to be \(\{1\}^{c \log n}\) and sets its \texttt{BeenCheated} flag to \texttt{true}, 
since such nodes are considered to have deviated from the algorithm.
Similarly, if \(i\) deviated in a manner that can be detected by some of its neighbors(this includes the case where \(i\) fails to broadcast 
\(r_{\mathrm{opp}(i) \rightarrow i}\) because it did not receive it from 
\(\mathrm{opp}(i)\)), 
then it sets its own rank to \(\{1\}^{c \log n}\) in order to stay consistent with neighbors that detect \(i\)'s deviation, and avoid a potential \(-\infty\) payoff.

Next, node \(i\) checks whether any neighbor has already output \(1\) or \(\bot\) in the current iteration; 
if so, it outputs \(\bot\) in every subsequent round in which it remains \texttt{Undecided}. . 
Otherwise, it compares its rank with all its undecided neighbors. 
If \(R_i < R_j\) for all such neighbors \(j\), 
then node \(i\) outputs \(1\), joining the MIS and terminates.
\item \textbf{Round 5: Reacting to Neighbor Decisions.}
Each node \(i\) checks whether any undecided neighbor \(j\) has output \(\bot\). 
If such a node exists, \(i\) outputs \(\bot\) in every subsequent round in which it remains \texttt{Undecided}.

Else, if any  neighbor \(j\) has output \(1\) in previous round, 
then:
\begin{itemize}
    \item If, for all such \(j\) with output \(1\), it holds that \(R_j < R_i\), 
    then node \(i\) outputs \(0\) in every subsequent round while it remains undecided, 
    unless \texttt{BeenCheated} is \texttt{true}, in which case it outputs \(\bot\) in every subsequent round while it remains undecided..
    \item Otherwise, if there exists a node \(j\) with \(R_j > R_i\) that still outputs \(1\), 
    then \(i\) outputs \(\bot\) in every subsequent round while it remains undecided.
\end{itemize}

If there are no \texttt{Undecided} neighbors, \(i\) outputs \(1\). otherwise, \(i\) proceeds to next iteration.
\end{itemize}

\noindent
The distributed strategy algorithm described above is formalized in Algorithm \ref{alg:crypto_dev}. Since the description is quite lengthy, we have modularized it into several algorithms, each describing one round of the iteration, presented in Algorithms \ref{alg:opp_selection} -- \ref{alg:react_to_nbrs}. We then proceed to prove the main theorem, Theorem~\ref{thm:crypto_main}, which demonstrates that the proposed algorithm is rationality resilient.

\begin{algorithm}
\caption{ \label{alg:crypto_dev}Rank based MIS algorithm for Node \( i \)}
\SetKwBlock{Init}{Initialization}{}
\Init{
    \texttt{UndecidedNeighbors} \( \gets N(i) \);
    \texttt{BeenCheated} \( \gets \) false;
    \texttt{iteration\_number} \( \gets 0 \)
}
\If{\(N(i) = \emptyset\)}{
   In every round where node \(i\) is still \textsc{Undecided}, Output \(1\)
}
\While{\(i\) is \emph{\textsc{Undecided}}}{

    \texttt{iteration\_number} \( \gets \) \texttt{iteration\_number} \(+ 1 \)\;
    
    \tcp{--- Round 1 ---}
    Run Algorithm Opponent Selection (Algorithm \ref{alg:opp_selection}) \;
    
    \tcp{--- Round 2 ---}
    Run Algorithm Generate Randomness (Algorithm \ref{alg:gen_rand_ranks})\;
    
    \tcp{--- Round 3 ---}
    Run Algorithm Forward Opponent's random input (Algorithm \ref{alg:share_opp_rand})\;
    
    \tcp{--- Round 4 ---}
    Run Algorithm Calculate Ranks of Neighbors (Algorithm \ref{alg:nbr_ranks})\;
    
    \tcp{--- Round 5 ---}
    Run Algorithm Reacting to Neighbor Decisions (Algorithm \ref{alg:react_to_nbrs})\;
    }

\end{algorithm}

\begin{algorithm}
\caption{\label{alg:opp_selection} Opponent Selection}
      Choose any \( k \in \) \texttt{UndecidedNeighbors} as the opponent \( \mathrm{opp}(i) \), and broadcast this choice to all neighbors.
\end{algorithm}

\begin{algorithm}
\caption{\label{alg:gen_rand_ranks} Generate Randomness}
 \If{\(i\) did not broadcast a valid ID in Round 1}{
    
    $R_i \gets \{1\}^{c\log n}$;
    
    \tcp{Node \(i\) deviated; as neighbors assume its rank is \(\{1\}^{c \log n}\), \(i\) also adopts this to stay consistent and avoid \(-\infty\) payoff.}

      Draw \(c\log {n}\) bits \( r_{i \rightarrow i} \in \{0,1\}^{c\log n} \) uniformly at random and broadcast to all the neighbors\;
        \ForEach{ \( j \in \) \emph{\texttt{UndecidedNeighbors}} which broadcasted \( \mathrm{opp}(j) = i \) in previous round}
        {
        Draw a random \(c\log {n}\) bit \( r_{i \rightarrow j} \in \{0,1\}^{c\log n}\) uniformly at random and
        broadcast signed message \( (\texttt{iteration\_number}, i,j, r_{i \rightarrow j}) \) to  all neighbors;
        }
    }
\ForEach{\(j \in \) \emph{\texttt{UndecidedNeighbors}} that did not broadcast an opponent in Round 1}{
    Node \(i\) sets \(R_j \gets \{1\}^{c\log n}\) and \texttt{BeenCheated} \(\gets\) true\tcc*{as \(j\) deviated}
}
\end{algorithm}

\begin{algorithm}
\caption{\label{alg:share_opp_rand} Forward Opponent's random input}

    \If{ \(i\) broadcasted ID of \(\mathrm{opp}(i)\) as its opponent in Round 1}{
        \If{ valid signed \(r_{\mathrm{opp}(i) \rightarrow i}\) is recieved from \(\mathrm{opp}(i)\) in Round 2}
        {
            \( R_{i} \gets \{r_{i\rightarrow i}\oplus r_{\mathrm{opp}(i)\rightarrow i}\}\)\;
            Broadcast message \( (\texttt{iteration\_number}, \mathrm{opp}(i), r_{\mathrm{opp}(i) \rightarrow i}) \) signed by \(\mathrm{opp}(i)\) to all neighbors 
        }
        \Else{
            \texttt{BeenCheated} \( \gets \) true\tcc*{\(i\) was cheated by \(\mathrm{opp}(i)\)}
            \( R_{i} \gets \{1\}^{c\log n} \)\tcc*{as \(i\) does not have signed message from \(\mathrm{opp}(i)\), its neighbors will assign this rank to \(i\)}
        }
    }

\end{algorithm}
\begin{algorithm}
\caption{\label{alg:nbr_ranks} Calculate Ranks of Neighbors}
  \If{any \( j \in\) \emph{\texttt{UndecidedNeighbors}} has output \(1\) or \( \bot \) }{
        From this point onward, disregard the rest of the algorithm and output \( \bot \) in every round where node \(i\) is still \textsc{Undecided};
    }
    \Else{
        Remove all \( j \in  \texttt{UndecidedNeighbors} \) that output \(0\) from   \texttt{UndecidedNeighbors};
    } 
    \ForEach{\( j \in \texttt{UndecidedNeighbors} \)}{
        \If{message \( (\texttt{iteration\_number}, \mathrm{opp}(j), r_{\mathrm{opp}(j) \rightarrow j})\) signed by \(\mathrm{opp}(j)\) is received from \(j\)}{
        Node \(i\) sets \( R_j \gets \{r_{\mathrm{opp}(j) \rightarrow j} \oplus r_{j \rightarrow j} \}\)\;    
        }
        \Else{
        Node \(i\) sets \(R_j \gets \{1\}^{c\log n}\), \texttt{BeenCheated} \(\gets\) true\tcc*{as \(j\) deviated}
        }
    }

    \If{\( R_{i} < R_j \) for all \( j \in\) \emph{\texttt{UndecidedNeighbors}}}{
        Output \( 1 \) (join MIS)\;
    }
\end{algorithm}
\begin{algorithm}
\caption{\label{alg:react_to_nbrs} Reacting to Neighbor Decisions}
  \If{any node in  \emph{\texttt{UndecidedNeighbors}} has output \( 1 \)}{
        \If{for all such \( j \) with output 1, \( R_j < R_{i} \)}{
            From this point onward, disregard the rest of the algorithm; if \texttt{BeenCheated} is \textbf{true}, then output \( \bot \), else output \(0\), in every round where node \(i\) is still \textsc{Undecided}.
        }
        \Else{
            From this point onward, disregard the rest of the algorithm and Output \(\bot\) in every round where node \(i\) is still \textsc{Undecided}
        }
    }
    \Else{
        Remove all \( j \in \) \texttt{UndecidedNeighbors} who output \( 0 \)\\
        \If{\emph{\texttt{UndecidedNeighbors}} is empty}{
            From this point onward, disregard the rest of the algorithm and Output \(1\) in every round where node \(i\) is still \textsc{Undecided}  
        }
    }

\end{algorithm}

\begin{theorem}\label{thm:crypto_main}
Algorithm \ref{alg:crypto_dev} is a 
Rationality Resilient Algorithm as defined in Definition \ref{def:RR}.
\end{theorem}
We prove Theorem~\ref{thm:crypto_main} in three parts (Theorems~\ref{thm:crypto_eq}, \ref{thm:crypto_correctness_termination}, and \ref{thm:crypto_ppi}). Among these, Theorem~\ref{thm:crypto_eq} requires intermediate results, which we establish through Lemmas~\ref{lem:crypto_been_cheated} and~\ref{lem:crypto_stick_to_algo}.

\begin{lemma}\label{lem:crypto_been_cheated}
Let the game be at an arbitrary history $h$ with $i\in P(h)$, where $h \in I_i \in \mathcal{I}_i$. 
Suppose that one of $i$'s \texttt{UndecidedNeighbors} (say, $j$) has its \texttt{BeenCheated} flag set to \texttt{true}.  
Then, regardless of the history $h$ and of $i$'s strategy from $h$ onward, if all the other nodes follow Algorithm~\ref{alg:crypto_dev} from $h$ onward, the utility of node $i$ will be at most~$0$.
\end{lemma}

\begin{proof}
    According to Algorithm \ref{alg:crypto_dev}, node \( j \) outputs \( 0 \) only in Round 5 of an iteration, and even then, only if its \texttt{BeenCheated} flag is \texttt{false}. Since the \texttt{BeenCheated} flag, once set to \texttt{true}, is never reset to \texttt{false}, a node with \texttt{BeenCheated} = \texttt{true} will never output \( 0 \). Therefore, node \( j \) will only output \( 1 \) or \( \bot \). In either case, the utility of \(i\) is \( 0 \).
\end{proof}

\begin{lemma}\label{lem:crypto_stick_to_algo}
    Let the game be at an arbitrary history $h$ with $ i\in P(h)$, where $h \in I_i \in \mathcal{I}_i$, and let \( h_i \) be the local history of node \( i \). 
    Then, provided that all the undecided nodes follow Algorithm~\ref{alg:crypto_dev} from this point onward, the utility of node \( i \) is at least \( 0 \).
\end{lemma}

\begin{proof}
    A node \( i \) receives a utility less than zero (i.e., \( -\infty \)) only in two cases:
    \begin{enumerate}
        \item \( i \) outputs \( 0 \), and none of its neighbors also outputs \( 1 \), or
        \item \( i \) outputs \( 1 \), and at least one of its neighbors also outputs \( 1 \).
    \end{enumerate}

    \textbf{Case 1} is not possible under the algorithm: a node is instructed to output \( 0 \) only in Round 4, and only after observing at least one neighbor output \( 1 \). Hence, with any history where the node \(i\) is undecided in Round 3 outputs \( 0 \), it must be the case that a neighbor has already output \( 1 \). Therefore, the first case cannot occur.

    \textbf{Case 2} is also prevented by the algorithm: a node outputs \( 1 \) only in Round 3 and only if:
    \begin{itemize}
        \item No neighbor has output \( 1 \) or \( \bot \) in the previous rounds, and
        \item \( i \) has the least rank compared to all its undecided Neighbors
    \end{itemize}
    Because all Undecided neighbors follow the algorithm, the nodes with rank less than \(i\) will not output \( 1 \) -- they will either output \( 0 \) or \(\bot\) in later Rounds.

    Therefore, following the algorithm ensures that node \( i \)'s utility is never less than \( 0 \).
\end{proof}

\begin{theorem}\label{thm:crypto_eq}
Algorithm \ref{alg:crypto_dev} constitutes a \EQ.
\end{theorem}
\begin{proof}
Consider an arbitrary history $h$ with $i \in P(h)$. Let $h_i$ denote the local history of node $i$, which implies $h \in I_i(h_i) \in \mathcal{I}_i$.  We show that the continuation strategy prescribed for $i$ by Algorithm~\ref{alg:crypto_dev} from information set $I_i(h_i)$ is sequentially rational irrespective of the \(i\)'s belief over $I_i(h_i)$, i.e., it gives the maximum expected payoff when all other nodes also follow Algorithm~\ref{alg:crypto_dev} from $h$. It follows that the strategy constitutes a sequential equilibrium for any belief system of \(i\). and it follows that the strategy constitutes a  \EQ.
As we mentioned in Section~\ref{sec:model}, we do not consider internal deviations; this indirectly means that node \(i\) knows exactly which line/round of the algorithm applies at \(h_i\).
We distinguish deviations as
\begin{itemize}
  \item \emph{Undetectable deviations:} for example, choosing internal randomness\(r_{i\rightarrow i},r_{i\rightarrow j}\) with nonuniform
probabilities. Since every sequence of ranks is still possible under the prescribed
uniform strategy (albeit with small probability), neighbors can never rule out a deviation
with certainty, even after many observations. Similarly, voluntarily outputting \(0\) or
\(\bot\) earlier than prescribed may also be indistinguishable, since it could be a valid
response to another neighbor having output \(1\), which remaining neighbors may not
observe.
  \item \emph{Detectable deviations:} e.g., failing to send a required/valid message,
  sending inconsistent messages, or outputting \(1\) when not eligible. Such deviations
  are caught by at least one neighbor, which sets its \texttt{BeenCheated} flag to \(1\).
\end{itemize}

\noindent
\textbf{Case analysis.}

\smallskip
\emph{Case A: At \(h\), \(i\) has a decided neighbor with output \(1\) or \(\bot\).}
Then the maximum attainable payoff for \(i\) is \(0\). By Lemma~\ref{lem:crypto_stick_to_algo},
following the algorithm yields a payoff of at least \(0\). Hence, no profitable deviation exists,
so the prescribed action is a best response.

\smallskip
\emph{Case B: At \(h\), \(i\) do not have a decided neighbor with output \(1\) or \(\bot\) and there exists an
undecided neighbor that \(i\) has cheated.}
If the remaining undecided neighbors follow the algorithm, Lemma~\ref{lem:crypto_been_cheated}
implies the most \(i\) can obtain is \(0\); following the algorithm also yields payoff \(0\).
Thus, no profitable deviation exists.

\smallskip
\emph{Case C: At \(h\), \(i\) do not have a decided neighbor with output \(1\) or \(\bot\) and \(i\) has not cheated
any undecided neighbor.}
\begin{itemize}
  \item \emph{Detectable deviation by \(i\):} Some neighbor detects it and sets
  \texttt{BeenCheated}\(=1\). By Lemma~\ref{lem:crypto_been_cheated}, the best payoff \(i\) can then secure is \(0\).
  By Lemma~\ref{lem:crypto_stick_to_algo}, sticking to the algorithm yields payoff at most \(0\) and at least \(0\),
  so no profitable deviation exists.
  \item \emph{Undetectable deviation by \(i\):} (1) Changing RPS mixing probabilities does not
  improve \(i\)'s win probability against neighbors who randomize uniformly; \(i\)'s chance
  of joining the MIS is unchanged, so expected payoff is unchanged. (2) Voluntarily outputting
  \(0\) or \(\bot\) can yield at most \(0\), whereas Lemma~\ref{lem:crypto_stick_to_algo} guarantees
  sticking yields payoff at least \(0\). Hence, no profitable deviation exists.
\end{itemize}

Therefore, we have shown that for every history \(h\) with \(i\in P(h)\),
the prescribed strategy for \(i\) based on \(h_i= \mathrm{proj}_i(h)\) gives maximum expected utility possible from \(h\) when all other node are following the algorithm from \(h\), which proves the theorem.
\end{proof}

We have shown that Algorithm \ref{alg:crypto_dev} forms a \EQ. To complete the proof of Theorem~\ref{thm:RPS_main_theorem}, we now assume that all nodes follow the algorithm without deviation. Under this honest execution, we analyze the simplified version of the algorithm (with deviation-handling logic removed) and proceed to prove its correctness, termination, and nonzero inclusion probability.

\subsection{Strategy Algorithm without Deviations}
Since we have established in Theorem~\ref{thm:crypto_eq} that Strategy Algorithm~\ref{alg:crypto_dev} constitutes a Belief-Independent Sequential Equilibrium, no node has a unilateral incentive to deviate. When no node deviates from the Strategy Algorithm~\ref{alg:crypto_dev}, its execution is equivalent to running Algorithm~\ref{alg: crypto_nodev}, which is simply Strategy Algorithm~\ref{alg:crypto_dev} with the deviation handling logic removed.

\begin{algorithm}
\caption{\label{alg: crypto_nodev} Rank based MIS strategy algorithm for node \(i\) under no deviations}
\vspace{1em}

\SetKwBlock{Init}{Initialization}{}
\Init{
    \texttt{UndecidedNeighbors} \( \gets N(i) \);
    \texttt{iteration\_number} \( \gets 0 \)
}
\If{\(N(i) = \emptyset\)}{
    Output \(1\)
}
\While{\(i\) is \emph{\textsc{Undecided}}}{
    \texttt{iteration\_number} \( \gets \texttt{iteration\_number} + 1 \)\;

    \tcc{Round 1: Opponent Selection}    
   Choose any \( k \in \) \texttt{UndecidedNeighbors} as the opponent \( \mathrm{opp}(i) \), and broadcast this choice to all neighbors.

    \tcc{Round 2: Priority Computation}
    
    Draw random \(c\log {n}\) bit  \( r_{i \rightarrow i} \in \{0,1\}^{c\log n} \) uniformly at  random and broadcast to all the neighbors\;  
        \ForEach{ \( j \in\) \emph{\texttt{UndecidedNeighbors}} which broadcasted \( \mathrm{opp}(j) = i \) in previous round}
        {
        Draw a random \(c\log {n}\) bit \( r_{i \rightarrow j} \in \{0,1\}^{c\log n}\) uniformly at random and
        broadcast signed message \( (\texttt{iteration\_number}, i, r_{i \rightarrow j}) \) to  all neighbors;
        }
        
    \tcc{Round 3: Forward Opponent's random input}

       \( R_{i} \gets \{r_{i\rightarrow i}\oplus r_{\mathrm{opp}(i)\rightarrow i}\}\)
        \newline
        Broadcast message \( (\texttt{iteration\_number}, \mathrm{opp}(i), r_{\mathrm{opp}(i) \rightarrow i}) \) signed by \(\mathrm{opp}(i)\) to all neighbors 
        
        \tcc{Round 4: Calculate Ranks of Neighbors} 
        
    \ForEach{\( j \in \) \emph{\texttt{UndecidedNeighbors}}}{
       
        \( R_j \gets \{r_{\mathrm{opp}(j) \rightarrow j} \oplus r_{j \rightarrow j} \}\)\; 
       
    }

    \If{\( R_{i} < R_j \) for all \( j \in\) \emph{\texttt{UndecidedNeighbors}}}{
        Output \( 1 \) (join MIS)\;
    }

    \tcc{Round 5: Reacting to Neighbor Decisions}
  \If{any \( j \in\) \emph{\texttt{UndecidedNeighbors}} has output \( 1 \)}{
       Output \(0\)
        }
    \Else{
        Remove all \( j \in \) \texttt{UndecidedNeighbors} who output \( 0 \)\;
        \If{\emph{\texttt{UndecidedNeighbors}} is empty}{
       Output \(1\) 
        }
    }
}
\end{algorithm}

We now proceed to prove the remaining claims of Theorem~\ref{thm:crypto_main} on this deviation-free execution: namely, correctness (that the output forms a valid MIS), termination (with high probability), and positive inclusion probability (every node has a nonzero chance of joining the MIS).

\begin{theorem}[Correctness and Termination] \label{thm:crypto_correctness_termination}
Algorithm \ref{alg: crypto_nodev} terminates in \( O(\log n) \) rounds with high probability, and the set of nodes that output \(1\) at termination forms a valid Maximal Independent Set (MIS).
\end{theorem}

\begin{proof}
Algorithm \ref{alg: crypto_nodev} simulates the algorithm by Métivier et al.~\cite{Metivier} through verifiable random priority generation and decision-making rounds.

Each iteration of Algorithm \ref{alg: crypto_nodev} consists of five communication rounds. This five-round process mirrors one round of the algorithm by Métivier et al.~\cite{Metivier}. Since that algorithm terminates in \( O(\log n) \) rounds with high probability and outputs a correct MIS, our algorithm also inherits these guarantees. Thus, Algorithm~\ref{alg: crypto_nodev} terminates in \( 5 \times O(\log n) \) iterations (i.e., \( O(\log n) \) rounds asymptotically). However, an additional argument is needed: in the original algorithm by Métivier et al.~\cite{Metivier}, the probability that two nodes pick the same rank is negligible due to the use of real-valued randomness. In contrast, in our algorithm, each node generates a rank using \( c \log n \) random bits, so the probability that any two nodes pick the same rank is non-negligible.

In each round, every node generates a rank consisting of \( c \log n \) random bits. The total number of possible ranks is \( 2^{c \log n} = n^c \). By the union bound, the probability that any two nodes pick the same rank in a given round is at most
\(
\binom{n}{2} \cdot \frac{1}{n^c} = O\left(\frac{1}{n^{c - 2}}\right).
\)
When this process is repeated independently for \( O (
\log n) \) rounds, the probability that any two nodes collide in at least one of the rounds is bounded by
\(
O\left( \frac{\log n}{n^{c - 2}} \right),
\)
which is negligible for any sufficiently large constant \( c > 2 \). Thus, with high probability, no two nodes ever pick the same rank across all rounds.

Since the original algorithm with real-valued ranks succeeds with probability \( 1 - 1/\mathrm{poly}(n)\), the bit-based version succeeds with probability at least
\(
1 - 1/\mathrm{poly}(n)\).

\end{proof}

\begin{theorem}[Positive Probability of MIS Inclusion]\label{thm:crypto_ppi}
Under Algorithm \ref{alg: crypto_nodev}, every node has a strictly positive probability of eventually outputting \(1\) (i.e., joining the MIS) in each round in which none of its neighbors has output \(1\).
\end{theorem}
\begin{proof}
Consider an undecided node \(i\). Its priority is computed as
\(
R_i = r_{i \rightarrow i} \oplus r_{\mathrm{opp}(i) \rightarrow i},
\)
where both random values are drawn independently and uniformly from 
\(\{0,1\}^{c \log n}\). The priorities of all neighbors are defined analogously, so that they are independent and identically distributed over \(\{0,1\}^{c \log n}\).

Hence, the probability that \(R_i\) is strictly smaller than the priorities of all its active neighbors is \(1/d_i\) where \(d_i\) is the number of undecided neighbors of \(i\). Therefore, as long as none of its neighbors has the output \(1\), node \(i\) always has a strictly positive probability of joining the MIS.

\end{proof}

\section{Conclusion and Future Work}
\label{sec:conclusion}
We proposed a model that incorporates rational behavior in distributed message passing algorithms. We considered a utility function, where nodes are incentivized to compute a correct solution and prefer to be included in the MIS. We designed two algorithms that are resilient to rational behavior: no node has an incentive to deviate from the prescribed algorithm. Additionally, assuming that no node deviates from the algorithm, we can guarantee correctness, termination, and positive probability of inclusion.

However, each algorithm comes with certain drawbacks. The rank-based strategy algorithm relies on cryptographic assumptions, which implicitly restrict agents to bounded computational power -- a non-standard assumption in the context of distributed algorithms. Moreover, this algorithm may use large messages; a node that becomes the opponent of $k$ neighbors will send \(O(k \log n)\)-bit messages over each incident edge. On the other hand, the RPS-based strategy algorithm avoids cryptographic assumptions and sends at most \(O(\log n)\)-bit messages over each edge per round. But its round complexity has exponential dependence on the maximum degree of the network.

While our algorithms are resilient to unilateral deviations, they do not prevent coordination or implicit collusion among multiple nodes. For example, once a node realizes that it cannot join the MIS (because one of its neighbors has already output \(1\)), it may still have an incentive to continue participating in the algorithm in a way that strategically reduces the chances of its neighbors, thereby improving the inclusion probability of the nodes with which it is colluding. Additionally, the two algorithms presented are somewhat sensitive to the assumptions about the underlying model and utility structure. If the utilities are perturbed even slightly, our equilibrium guarantees may not hold. For instance, if the utility of a node in a locally valid solution where it belong to the MIS is negative, our algorithms no longer ensure equilibrium. 
An important direction for future work, therefore, is to design algorithms that remain robust under a broader class of utility functions, collusion, or to establish \emph{impossibility results} that characterize the precise limitations of such robustness.

\appendix
\section{Useful Tail Inequalities}
\label{sec:gt-prelims}
Following are two tail inequalities we use in our analysis of the running time of the RPS based MIS algorithm.

\begin{lemma}[Reverse Markov Inequality]\label{Markov}
Let \( Y \) be a real-valued random variable such that \( Y \leq a \) almost surely, for some constant \( a \in \mathbb{R} \), and let \( d < \mathbb{E}[Y] \). Then,
\[
\Pr\left(Y \geq d \right) \geq \frac{\mathbb{E}[Y] - d}{a - d}.
\]
\end{lemma}

\begin{lemma}[Chernoff Bound: Lower Tail]\label{chernoff}
Let \( X_1, X_2, \dots, X_n \) be independent Bernoulli random variables taking values in \(\{0,1\}\), and let
\(
X = \sum_{i=1}^n X_i\), and \(\mu = \mathbb{E}[X]
\).
Then, for any \( 0 < \delta < 1 \), the following inequality holds:
\[
\Pr(X \leq (1 - \delta)\mu) \leq \exp\left( -\frac{\delta^2 \mu}{2} \right)
\]
\end{lemma}

\end{document}